%% file: evoteid2020-twocandidate.tex
\documentclass[runningheads]{llncs}

\usepackage{amsmath}
\usepackage{amsfonts}
\usepackage{amssymb}
\usepackage{booktabs}
\usepackage{footmisc}
\usepackage{graphicx}
\usepackage[colorlinks]{hyperref}
\usepackage{threeparttable}

\newcommand{\appendixref}[1]{\hyperref[#1]{Appendix~\ref*{#1}}}

\newcommand{\ignore}[1]{}

\DeclareMathOperator{\E}{\mathbb{E}}

\title{A Unified Evaluation of Two-Candidate Ballot-Polling Election
Auditing Methods\thanks{An earlier version of this manuscript was published in:
Krimmer R. et al.\ (eds) Electronic Voting. E-Vote-ID 2020. Lecture Notes in
Computer Science, vol 12455. Springer, Cham.
\url{https://doi.org/10.1007/978-3-030-60347-2_8}.  The current version
includes some corrections to the main text and two extensive technical
appendices.}}

\titlerunning{Evaluation of Two-Candidate Ballot-Polling Election Auditing
Methods}

\author{
Zhuoqun Huang    \inst{1}                                  \and
Ronald L. Rivest \inst{2}   \orcidID{0000-0002-7105-3690}  \and
Philip B. Stark  \inst{3}   \orcidID{0000-0002-3771-9604}  \and
Vanessa Teague   \inst{4,5} \orcidID{0000-0003-2648-2565}  \and
Damjan Vukcevic  \inst{1,6} \orcidID{0000-0001-7780-9586}}

\authorrunning{Huang et al.}

\institute{
School of Mathematics and Statistics,
University of Melbourne, Parkville, Australia
\and
Computer Science \& Artificial Intelligence Laboratory, Massachusetts Institute
of Technology, USA
\and
Department of Statistics, University of California, Berkeley, USA
\and
Thinking Cybersecurity Pty.\ Ltd.
\and
College of Engineering and Computer Science, Australian National University
\and
Melbourne Integrative Genomics,
University of Melbourne, Parkville, Australia \\
\email{damjan.vukcevic@unimelb.edu.au}}

\begin{document}

\maketitle

\begin{abstract}
Counting votes is complex and error-prone.  Several statistical methods have
been developed to assess election accuracy by manually inspecting randomly
selected physical ballots.  Two `principled' methods are risk-limiting audits
(RLAs) and Bayesian audits (BAs).  RLAs use frequentist statistical inference
while BAs are based on Bayesian inference.  Until recently, the two have been
thought of as fundamentally different.

We present results that unify and shed light upon `ballot-polling' RLAs and BAs
(which only require the ability to sample uniformly at random from all cast
ballot cards) for two-candidate plurality contests, which are building blocks
for auditing more complex social choice functions, including some preferential
voting systems.  We highlight the connections between the methods and explore
their performance.

First, building on a previous demonstration of the mathematical equivalence of
classical and Bayesian approaches, we show that BAs, suitably calibrated, are
risk-limiting.  Second, we compare the efficiency of the methods across a wide
range of contest sizes and margins, focusing on the distribution of sample
sizes required to attain a given risk limit.  Third, we outline several ways to
improve performance and show how the mathematical equivalence explains the
improvements.

\end{abstract}


\section{Introduction}

Even if voters verify their ballots and the ballots are kept secure, the
counting process is prone to errors from malfunction, human error, and
malicious intervention.  For this reason, the US National Academy of Sciences
\cite{securing2018} and the American Statistical
Association\footnote{\url{https://www.amstat.org/asa/files/pdfs/POL-ASARecommendsRisk-LimitingAudits.pdf}}
have recommended the use of risk-limiting audits to check reported election
outcomes.

The simplest audit is a manual recount, which is usually expensive and
time-consuming.  An alternative is to examine a random sample of the ballots
and test the result statistically.  Unless the margin is narrow, a sample far
smaller than the whole election may suffice.  For more efficiency, sampling can
be done adaptively: stop when there is strong evidence supporting the reported
outcome \cite{stark08a}.

Risk-limiting audits (RLAs) have become the audit method recommended for use in
the USA.  Pilot RLAs have been conducted for more than 50~elections in 14~US
states and Denmark since 2008.  Some early pilots are discussed in a report
from the California Secretary of State to the US Election Assistance
Commission.\footnote{\url{https://votingsystems.cdn.sos.ca.gov/oversight/risk-pilot/final-report-073014.pdf}}
In 2017, the state of Colorado became the first to complete a statewide
RLA.\footnote{\url{https://www.denverpost.com/2017/11/22/colorado-election-audit-complete/}}
The defining feature of RLAs is that, if the reported outcome is incorrect,
they have a large, pre-specified minimum probability of discovering this and
correcting the outcome.  Conversely, if the reported outcome is correct, then
they will eventually certify the result.  This might require only a small
random sample, but the audit may lead to a complete manual tabulation of the
votes if the result is very close or if tabulation error was an appreciable
fraction of the margin.

RLAs exploit frequentist statistical hypothesis testing.  There are by now more
than half a dozen different approaches to conducting RLAs \cite{shangrla}.
Election audits can also be based on Bayesian inference
\cite{RivestAElections}.

With so many methods, it may be hard to understand how they relate to each
other, which perform better, which are risk-limiting, etc.  Here, we review and
compare the statistical properties of existing methods in the simplest case: a
two-candidate, first-past-the-post contest with no invalid ballots.  This
allows us to survey a wide range of methods and more clearly describe the
connections and differences between them.  Most real elections have more than
two candidates, of course.  However, the methods designed for this simple
context are often adapted for more complex elections by reducing them into
pairwise contests (see below for further discussion of this point).  Therefore,
while we only explore a simple scenario, it sheds light on how the various
approaches compare, which may inform future developments in more complex
scenarios.  There are many other aspects to auditing that matter greatly in
practice, we do not attempt to cover all of these but we comment on some below.

For two-candidate, no-invalid-vote contests, we explain the connections and
differences among many audit methods, including frequentist and Bayesian
approaches.  We evaluate their efficiency across a range of election sizes and
margins.  We also explore some natural extensions and variations of the
methods.  We ensure that the comparisons are `fair' by numerically calibrating
each method to attain a specified risk limit.

We focus on \emph{ballot-polling audits}, which involve selecting ballots at
random from the pool of cast ballots.  Each sampled ballot is interpreted
manually; those interpretations comprise the audit data.  (Ballot-polling
audits do not rely on the voting system's interpretation of ballots, in
contrast to \emph{comparison audits}.)

\textbf{Paper outline:}
\autoref{sec:definitions} provides context and notation.
\autoref{sec:auditing_methods} sketches the auditing methods we consider and
points out the relationships among them and to other statistical methods.
\autoref{sec:evaluating_methods} explains how we evaluate these methods.  Our
benchmarking experiments are reported in \autoref{sec:results}.  We finish with
a discussion and suggestions for future work in \autoref{sec:discussion}.


\section{Context and notation: two-candidate contests}
\label{sec:definitions}

We consider contests between two candidates, where each voter votes for exactly
one candidate.  The candidate who receives more votes wins.  Ties are possible
if the number of ballots is even.

Real elections may have invalid votes, for example, ballots marked in favour of
both candidates or neither; for multipage ballots, not every ballot paper
contains every contest.  Here we assume every ballot has a valid vote for one
of the two candidates.  See \autoref{sec:discussion}.

Most elections have more than two candidates and can involve complex algorithms
(`social choice functions') for determining who won.  A common tactic for
auditing these is to reduce them to a set of pairwise contests such that
certifying all of the contests suffices to confirm the reported outcome
\cite{Lindeman2012BRAVO:Outcomes,raire,shangrla}.  These contests can be
audited simultaneously using methods designed for two candidates that can
accommodate invalid ballots, which most of the methods considered below do.
Therefore, the methods we evaluate form the building blocks for many of the
more complex methods, so our results are more widely relevant.

We do not consider \emph{stratified audits}, which account for ballots cast
across different locations or by different voting methods within the same
election.

\subsection{Ballot-polling audits for two-candidate contests}

We use the terms `ballot' and `ballot card' interchangeably, even though
typical ballots in the US consist of more than one card (and the distinction
does matter for workload and for auditing methods).  We consider unweighted
\emph{ballot-polling} audits, which require only the ability to sample
uniformly at random from all ballot cards.

The sampling is typically sequential.  We draw an initial sample and assess the
evidence for or against the reported outcome.  If there is sufficient evidence
that the reported outcome is correct, we stop and `certify' the winner.
Otherwise, we inspect more ballots and try again, possibly continuing to a full
manual tabulation.  At any time, the auditor can chose to conduct a full hand
count rather than continue to sample at random.  That might occur if the work
of continuing the audit is anticipated to be higher than that of a full hand
count or if the audit data suggest that the reported outcome is wrong.  One
reasonable rule is to set a maximum sample size (number of draws, not
necessarily the number of distinct ballots) for the audit; if the sample
reaches that size but the outcome has not been confirmed, there is a full
manual tabulation.  The outcome according to that manual tabulation becomes
official.

There are many choices to be made, including:
\begin{description}
\item[How to assess evidence.] Each stage involves calculating a statistic from
    the sample.  What statistic do we use?  This is one key difference amongst
    auditing methods, see \autoref{sec:auditing_methods}.

\item[Threshold for evidence.]  The decision of whether to certify or keep
    sampling is done by comparing the statistic to a reference value.  Often
    the value is chosen such that it limits the probability of certifying the
    outcome if the outcome is wrong, i.e.\ limits the risk (see below).

\item[Sampling with or without replacement.]  Sampling may be done with or
    without replacement.  Sampling without replacement is more efficient;
    sampling with replacement often yields simpler mathematics.  The difference
    in efficiency is small unless a substantial fraction (e.g.\ 20\% or more)
    of the ballots are sampled.

\item[Sampling increments.]  By how much do we increase the sample size if the
    current sample does not confirm the outcome?  We could enlarge the sample
    one ballot at a time, but it is usually more efficient to have larger
    `rounds'.  The methods described here can accommodate rounds of any size.
\end{description}
We assume that the auditors read votes correctly, which generally requires
retrieving the correct ballots and correctly applying legal rules for
interpreting voters' marks.

\subsection{Notation}

Let $X_1, X_2, \dots \in \{0, 1\}$ denote the sampled ballots, with $X_i = 1$
representing a vote in favour of the reported winner and $X_i = 0$ a vote for
the reported loser.

Let $n$ denote the number of (not necessarily distinct) ballots sampled at a
given point in the audit, $m$ the maximum sample size (i.e.\ number of draws)
for the audit, and $N$ the total number of cast ballots.  We necessarily have
$n \leqslant m$ and if sampling without replacement we also have $m \leqslant
N$.

Each audit method summarizes the evidence in the sample using a statistic of
the form $S_n(X_1, X_2, \dots, X_n, n, m, N)$.  For brevity, we suppress $n$,
$m$ and $N$ in the notation.

Let $Y_n = \sum_{i = 1}^n X_i$ be the number of sampled ballots that are in
favour of the reported winner.  Since the ballots are by assumption
exchangeable, the statistics used by most methods can be written in terms of
$Y_n$.

Let $T$ be the \emph{true} total number of votes for the winner and $p_T = T /
N$ the true proportion of such votes.  Let $p_r$ be the \emph{reported}
proportion of votes for the winner.  We do not know $T$ nor $p_T$, and it is
not guaranteed that $p_r \simeq p_T$.

For sampling with replacement, conditional on $n$, $Y_n$ has a binomial
distribution with parameters $n$ and $p_T$.  For sampling without replacement,
conditional on $n$, $Y_n$ has a hypergeometric distribution with parameters
$n$, $T$ and $N$.

\subsection{Risk-limiting audits as hypothesis tests}
\label{sec:hypothesis-tests}

Risk-limiting audits amount to statistical hypothesis tests.  The null
hypothesis $H_0$ is that the reported winner(s) did \emph{not} really win.  The
alternative $H_1$ is that the reported winners really won.  For a single-winner
contest,
\begin{align*}
    H_0\colon & p_T \leqslant \tfrac{1}{2}, \quad \text{(reported winner is
                                                         false)} \\
    H_1\colon & p_T    >      \tfrac{1}{2}. \quad \text{(reported winner
                                                         is true)}
\end{align*}
If we reject $H_0$, we certify the election without a full manual tally.
The \emph{certification rate} is the probability of rejecting $H_0$.
Hypothesis tests are often characterized by their \emph{significance level}
(false positive rate) and \emph{power}.  Both have natural interpretations in
the context of election audits by reference to the certification rate.
The power is simply the certification rate when $H_1$ is true.  Higher power
reduces the chance of an unnecessary recount.
A false positive is a \emph{miscertification}: rejecting $H_0$ when in fact it
is true.  The probability of miscertification depends on $p_T$ and the audit
method, and is known as the \emph{risk} of the method.  In a two-candidate
plurality contest, the maximum possible risk is typically attained when $p_T =
\frac{1}{2}$.

For many auditing methods we can find an upper bound on the maximum possible
risk, and can also set their evidence threshold such that the risk is limited
to a given value.  Such an upper bound is referred to as a \emph{risk limit},
and methods for which this is possible are called \emph{risk-limiting}.  Some
methods are explicitly designed to have a convenient mechanism to set such a
bound, for example via a formula.  We call such methods \emph{automatically
risk-limiting}.

Audits with a sample size limit $m$ become full manual tabulations if they have
not stopped after drawing the $m$th ballot.  Such a tabulation is assumed to
find the correct outcome, so the power of a risk-limiting audit is 1.  We use
the term `power' informally to refer to the chance the audit stops after
drawing $m$ or fewer ballots.


\section{Election auditing methods}
\label{sec:auditing_methods}

We describe Bayesian audits in some detail because they provide a mathematical
framework for many (but not all) of the other methods.  We then describe the
other methods, many of which can be viewed as Bayesian audits for a specific
choice of the prior distribution.  Some of these connections were previously
described by \cite{Vora2019Risk-LimitingElections}.  These connections can shed
light on the performance or interpretation of the other methods.  However, our
benchmarking experiments are frequentist, even for the Bayesian audits (for
example, we calibrate the methods to limit the risk).

\autoref{tab:summary_methods} lists the methods described here; the parameters
of the methods are defined below.

\begin{table}[t]
\centering
\begin{threeparttable}
\caption{\textbf{Summary of auditing methods.}  The methods in the first part
of the table are benchmarked in this report.}
\label{tab:summary_methods}
\smallskip
\begin{tabular}{lcc}
\toprule
\textbf{Method}                     &
\textbf{Quantities to set} \quad{}  &
\textbf{Automatically risk-limiting} \\
\midrule
Bayesian             & $f(p)$     & --- \\
Bayesian (risk-max.) & $f(p), \text{for } p > 0.5$
                                  & \checkmark \\
BRAVO                & $p_1$      & \checkmark \\
MaxBRAVO             & None       & ---        \\
ClipAudit            & None       & ---\tnote{\textdagger} \\
\midrule
KMart                & $g(\gamma)$\tnote{\textdaggerdbl}
                                  & \checkmark  \\
Kaplan--Wald         & $\gamma$   & \checkmark  \\
Kaplan--Markov       & $\gamma$   & \checkmark  \\
Kaplan--Kolmogorov   & $\gamma$   & \checkmark  \\
\bottomrule
\end{tabular}
\begin{tablenotes}
\item[\textdagger] Provides a pre-computed table for approximate risk-limiting
    thresholds
\item[\textdaggerdbl] Extension introduced here
\end{tablenotes}
\end{threeparttable}
\end{table}

\subsection{Bayesian audits}
\label{sec:bayesian}

Bayesian audits quantify evidence in the sample as a posterior distribution of
the proportion of votes in favour of the reported winner.  In turn, that
distribution induces a (posterior) probability that the outcome is wrong,
$\Pr(H_0 \mid Y_n)$, the \emph{upset probability}.

The posterior probabilities require positing a \emph{prior distribution}, $f$
for the reported winner's vote share $p$.  (For clarity, we denote the fraction
of votes for the reported winner by $p$ when we treat it as random for Bayesian
inference and by $p_T$ to refer to the actual true value.)

We represent the posterior using the posterior odds,
\[
   \frac{\Pr(H_1 \mid X_1, \dots, X_n)}
        {\Pr(H_0 \mid X_1, \dots, X_n)}
 = \frac{\Pr(X_1, \dots, X_n \mid H_1)}
        {\Pr(X_1, \dots, X_n \mid H_0)} \times
   \frac{\Pr(H_1)}
        {\Pr(H_0)}.
\]
The first term on the right is the \emph{Bayes factor} (BF) and the second is
the prior odds.  The prior odds do not depend on the data:  the information
from the data is in the BF.  We shall use the BF as the statistic, $S_n$.  It
can be expressed as,
\[
S_n = \frac{\Pr(X_1, \dots, X_n \mid H_1)}
           {\Pr(X_1, \dots, X_n \mid H_0)}
    = \frac{\int_{p    >      0.5} \Pr(Y_n \mid p) \, f(p) \, dp}
           {\int_{p \leqslant 0.5} \Pr(Y_n \mid p) \, f(p) \, dp}.
\]
The term $\Pr(Y_n \mid p)$ is the \emph{likelihood}.  The BF is similar to a
likelihood ratio, but the likelihoods are integrated over $p$ rather than
evaluated at specific values (in contrast to classical approaches, see
\autoref{sec:sprt}).
%

\subsubsection{Understanding priors.}

The prior $f$ determines the relative contributions of possible values of $p$
to the BF.  It can be continuous, discrete, or a combination of the two.  A
\emph{conjugate prior} is often used \cite{RivestAElections}, which has the
property that the posterior distribution is in the same family, which has
mathematical and practical advantages.  For sampling with replacement the
conjugate prior is beta (which is continuous), while for sampling without
replacement it is a beta-binomial (which is discrete).

Vora \cite{Vora2019Risk-LimitingElections} showed that a prior that places a
probability mass of 0.5 on the value $p = 0.5$ and the remaining mass on $(1/2,
1]$ is \emph{risk-maximizing}.  For such a prior, limiting the upset
probability to $\upsilon$ also limits the risk: for the specific type of
Bayesian audits considered by Vora \cite{Vora2019Risk-LimitingElections}, the
risk limit is $\upsilon$; however, for the Bayesian audits described here (see
below), the risk limit is $\upsilon / (1 - \upsilon) > \upsilon$.\footnote{%
These differences are described in more detail in \appendixref{app:bayes-rla}.
}

We explore several priors below, emphasizing a uniform prior (an example of a
`non-partisan prior' \cite{RivestAElections}), which is a special case within
the family of conjugate priors used here.

\subsubsection{Bayesian audit procedure.}

A Bayesian audit proceeds as follows.  At each stage of sampling, calculate
$S_n$ and then:
\begin{equation}
\label{eqn:bayesian_audit}
\left\{
\begin{array}{ll}
    \text{if } S_n    >      h,  &  \quad \text{terminate and certify,} \\
    \text{if } S_n \leqslant h,  &  \quad \text{continue sampling.} \\
\end{array}
\right.
\tag{*}
\end{equation}
If the audit does not terminate and certify for $n \leqslant m$, there is a
full manual tabulation of the votes.

The threshold $h$ is equivalent to a threshold on the upset probability:
$\Pr(H_0 \mid Y_n) < \upsilon$ corresponds to $h = \frac{1 -
\upsilon}{\upsilon} \frac{\Pr(H_0)}{\Pr(H_1)}$.  If the prior places equal
probability on the two hypotheses (a common choice), this simplifies to $h =
\frac{1 - \upsilon}{\upsilon}$.

\subsubsection{Interpretation.}

The upset probability, $\Pr(H_0 \mid Y_n)$, is \textbf{not} the risk, which is
$\Pr(\text{certify} \parallel p_T)$.  The procedure outlined above limits the
upset probability.  This is not the same as limiting the risk.  Nevertheless,
in the election context considered here, Bayesian audits are
risk-limiting,\footnote{This is a consequence of the fact that the risk of a
Bayesian audit is largest when $p_T = 0.5$, a fact that we can use to bound the
risk by choosing an appropriate value for the threshold.  The mathematical
details are shown in \appendixref{app:bayes-rla}.} with a risk limit that is
generally larger than the upset probability threshold.

For a given prior, sampling scheme, and risk limit $\alpha$, we can calculate a
value of $h$ for which the risk of the Bayesian audit with threshold $h$ is
bounded by $\alpha$.  For risk-maximizing priors, taking $h = 1 / \alpha$
(which is equivalent to a threshold of $\upsilon = \alpha / (1 + \alpha)$ on
the upset probability) yields an audit with risk limit $\alpha$.

\subsection{SPRT-based audits}
\label{sec:sprt}

The basic sequential probability ratio test (SPRT)
\cite{Wald1945SequentialHypotheses}, adapted slightly to suit the auditing
context here,\footnote{The SPRT allows rejection of either $H_0$ or $H_1$, but
we only allow the former here.  This aligns it with the broader framework for
election audits described earlier.  Also, we impose a maximum sample size, as
we do for the other methods.} tests the simple hypotheses
\begin{align*}
    H_0\colon & p_T = p_0, \\
    H_1\colon & p_T = p_1,
\end{align*}
using the likelihood ratio:
\[
\left\{
\begin{array}{ll}
    \text{if } S_n = \frac{\Pr(Y_n \parallel p_1)}
                          {\Pr(Y_n \parallel p_0)} > \frac{1}{\alpha},
                       &  \quad \text{terminate and certify (reject $H_0$),} \\
    \text{otherwise,}  &  \quad \text{continue sampling.} \\
\end{array}
\right.
\]
This is equivalent to \eqref{eqn:bayesian_audit} for a prior with point masses
of 0.5 on the values $p_0$ and $p_1$ with $h = 1 / \alpha$.  This procedure has
a risk limit of $\alpha$.

The test statistic can be tailored to sampling with or without replacement by
using the appropriate likelihood.  The SPRT has the smallest expected sample
size among all level $\alpha$ tests of these same hypotheses.  This optimality
holds only when no constraints are imposed on the sampling (such as a maximum
sample size).

The SPRT statistic is a nonnegative martingale when $H_0$ holds; Kolmogorov's
inequality implies that it is automatically risk-limiting.  Other
martingale-based tests are discussed in \autoref{sec:other-methods}.

The statistic from a Bayesian audit can also be a martingale, if the prior is
the true data generating process under $H_0$.  This occurs, for example, for a
risk-maximizing prior if $p_T = 0.5$.\footnote{Such a prior places all its mass
on $p = 0.5$ when $p \leqslant 0.5$.}

\subsubsection{BRAVO.}

In a two-candidate contest, BRAVO~\cite{Lindeman2012BRAVO:Outcomes} applies the
SPRT with:
\begin{align*}
    p_0 &= 0.5, \\
    p_1 &= p_r - \epsilon,
\end{align*}
where $\epsilon$ is a pre-specified small value for which $p_1 >
0.5$.\footnote{The SPRT can perform poorly when $p_T \in (p_0, p_1)$; taking
$\epsilon > 0$ protects against the possibility that the reported winner really
won, but not by as much as reported.} Because it is the SPRT, BRAVO has a risk
limit no larger than $\alpha$.

BRAVO requires picking $p_1$ (analogous to setting a prior for a Bayesian
audit).  The recommended value is based on the reported winner's share, but the
SPRT can be used with any alternative.  Our numerical experiments do not
involve a reported vote share; we simply set $p_1$ to various values.

\subsubsection{MaxBRAVO.}
\label{sec:max_bravo}

As an alternative to specifying $p_1$, we experimented with replacing the
likelihood, $\Pr(Y_n \parallel p_1)$, with the maximized likelihood,
$\max_{p_1} \Pr(Y_n \parallel p_1)$, leaving other aspects of the test
unchanged.  This same idea has been used in other contexts, under the name
MaxSPRT~\cite{Kulldorff2011ASurveillance}.  We refer to our version as
\emph{MaxBRAVO}.  Because of the maximization, the method is not automatically
risk-limiting, so we calibrate the stopping threshold $h$ numerically to attain
the desired risk limit, as we do for Bayesian audits.

\subsection{ClipAudit}

Rivest \cite{Rivest2017ClipAudit:Audit} introduces \emph{ClipAudit}, a method
that uses a statistic that is very easy to calculate, $S_n = (A_n - B_n) /
\sqrt{A_n + B_n}$, where $A_n = Y_n$ and $B_n = n - Y_n$.  Appoximately
risk-limiting thresholds for this statistic were given (found numerically),
along with formulae that give approximate thresholds.  We used ClipAudit with
the `best fit' formula \cite[equation~(6)]{Rivest2017ClipAudit:Audit}.

As far as we can tell, ClipAudit is not related to any of the other methods we
describe here, but $S_n$ is the test statistic commonly used to test the
hypothesis $H_0\colon p_T = 0.5$ against $H_1\colon p_T > 0.5$:
\[
S_n
= \frac{      A_n - B_n }
       {\sqrt{A_n + B_n}}
= \frac{Y_n - n + Y_n}{\sqrt{n}}
= \frac{Y_n / n - 0.5}{\sqrt{0.5 \times (1 - 0.5) / n}}
= \frac{\hat{p}_T - p_0}{\sqrt{p_0 \times (1 - p_0) / n}}.
\]

\subsection{Other methods}
\label{sec:other-methods}

Several martingale-based methods have been developed for the general problem of
testing hypotheses about the mean of a non-negative random variable.  SHANGRLA
exploits this generality to allow auditing of a wide class of
elections~\cite{shangrla}.  While we did not benchmark these methods in our
study (they are better suited for other scenarios, such as comparison audits,
and will be less efficient in the simple case we consider here), we describe
them here in order to point out some connections among the methods.

The essential difference between methods is in the definition of the statistic,
$S_n$.  Given the statistic, the procedure is the same: certify the election if
$S_n > 1 / \alpha$; otherwise, keep sampling.  All of the procedures can be
shown to have risk limit $\alpha$.

All the procedures involve a parameter $\gamma$ that prevents degenerate values
of $S_n$.  This parameter either needs to be set to a specific value or is
integrated out.

The statistics below that are designed for sampling without replacement depend
on the order in which ballots are sampled.  None of the other statistics (in
this section or earlier) have that property.

We use $t$ to denote the value of $\E(X_i)$ under the null hypothesis.  In the
two-candidate context discussed in this paper, $t = p_0 = 0.5$.

We have presented the formulae for the statistics a little differently to
highlight the connections among these methods.  For simplicity of notation, we
define $Y_0 = 0$.

\subsubsection{KMart.}

This method was described online under the name
\emph{KMart}\footnote{\url{https://github.com/pbstark/MartInf/blob/master/kmart.ipynb}}
and is implemented in SHANGRLA \cite{shangrla}.  There are two versions of the
test statistic, designed for sampling with or without
replacement,\footnote{When sampling without replacement, if we ever observe
$Y_n > Nt$ then we ignore the statistic and terminate the audit since $H_1$ is
guaranteed to be true.} respectively:
\[
S_n = \int_0^1 \prod_{i = 1}^n
      \left(\gamma \left[\frac{X_i}{t} - 1\right] + 1\right) d\gamma,
\text{ and }
S_n = \int_0^1 \prod_{i = 1}^n
      \left(\gamma \left[X_i \frac{\left(\frac{N - i + 1}{N}\right)}
                                  {t - \frac{1}{N} Y_{i - 1}} - 1\right] +
            1\right) d\gamma.
\]

This method is related to Bayesian audits for two-candidate contests: for
sampling with replacement and no invalid votes, we have shown that KMart is
equivalent to a Bayesian audit with a risk-maximizing prior that is uniform
over $p > 0.5$.\footnote{The mathematical details are shown in
\appendixref{app:kmart-bayes}.  \label{fn:kmartBayesian}}  The same analysis
shows how to extend KMart to be equivalent to using an arbitrary
risk-maximizing prior, by inserting an appropriately constructed weighting
function $g(\gamma)$ into the integrand.\footref{fn:kmartBayesian}

There is no direct relationship of this sort for the version of KMart that uses
sampling without replacement, since this statistic depends on the order the
ballots are sampled but the statistic for Bayesian audits does not.

\subsubsection{Kaplan--Wald.}

This method is similar to KMart but involves picking a value for $\gamma$
rather than integrating over $\gamma$ \cite{StarkTeague2014}.  The previous
proof\,\footref{fn:kmartBayesian} shows that for sampling with replacement,
Kaplan--Wald is equivalent to BRAVO with $p_1 = (\gamma + 1) / 2$; while for
sampling without replacement, there is no such relationship.

\subsubsection{Kaplan--Markov.}

This method applies Markov's inequality to the martingale $\prod_{i \leqslant
n} X_i / \E(X_i)$, where the expectation is calculated assuming sampling with
replacement \cite{stark2009}.  This gives the statistic
$
S_n = \prod_{i = 1}^n
\left(X_i + \gamma\right) /
\left(  t + \gamma\right)
$.

\subsubsection{Kaplan--Kolmogorov.}

This method is the same as Kaplan--Markov but with the expectation calculated
assuming sampling without replacement \cite{shangrla}.  This gives the
statistic
$
S_n = \prod_{i = 1}^n
\left[\left(X_i + \gamma\right)\left(\frac{N - i + 1}{N}\right)\right] /
\left[   t - \frac{1}{N} Y_{i - 1} + \frac{N - i + 1}{N} \gamma\right]
$.\footnote{As for KMart, if $Y_n > Nt$, the audit terminates: the null
hypothesis is false.}


\section{Evaluating auditing methods}
\label{sec:evaluating_methods}

We evaluated the methods using simulations; see the first part of
\autoref{tab:summary_methods}.

For each method, the termination threshold $h$ was calibrated numerically to
yield maximum risk as close as possible to 5\%.  This makes comparisons among
the methods `fair'.  We calibrated even the automatically risk-limiting
methods, resulting in a slight performance boost.  We also ran some experiments
without calibration, to quantify this difference.

We use three quantities to measure performance: maximum risk and `power',
defined in \autoref{sec:hypothesis-tests}, and the mean sample size.

\subsubsection{Choice of auditing methods.}

Most of the methods require choosing the form of statistics, tuning parameters,
or a prior.  Except where stated, our benchmarking experiments used sampling
without replacement.  Except where indicated, we used the version of each
statistic designed for the method of sampling used.  For example, we used a
hypergeometric likelihood when sampling without replacement.  For Bayesian
audits we used a beta-binomial prior (conjugate to the hypergeometric
likelihood) with shape parameters $a$ and $b$.  For BRAVO, we tried several
values of $p_1$.

The tests labelled `BRAVO' are tests of a method related to but not identical
to BRAVO, because there is no notion of a `reported' vote share in our
experiments.  Instead, we set $p_1$ to several fixed values to explore how the
underlying test statistic (from the SPRT) performs in different scenarios.

For MaxBRAVO and Bayesian audits with risk-maximizing prior, due to time
constraints we only implemented statistics for the binomial likelihood (which
assumes sampling with replacement).  While these are not exact for sampling
without replacement, we believe this choice has only a minor impact when $m \ll
N$ (based on our results for the other methods when using different
likelihoods).

For Bayesian audits with a risk-maximizing prior, we used a beta distribution
prior (conjugate to the binomial likelihood) with shape parameters $a$ and $b$.

ClipAudit only has one version of its statistic.  It is not optimized for
sampling without replacement (for example, if you sample \textbf{all} of the
ballots, it will not `know' this fact), but the stopping thresholds are
calibrated for sampling without replacement.

\subsubsection{Election sizes and sampling designs.}

We explored combinations of election sizes $N \in \{500, 1000, 5000, 10000,
20000, 30000\}$ and maximum sample sizes $m \in \{500, 1000, 2000, 3000\}$.
Most of our experiments used a sampling increment of 1 (i.e.\ check the
stopping rule after each ballot is drawn).  We also varied the sampling
increment (values in $\{2, 5, 10, 20, 50, 100, 250, 500, 1000, 2000\}$) and
tried sampling with replacement.

\subsubsection{Benchmarking via dynamic programming.}

We implemented an efficient method for calculating the performance measures
using dynamic programming.\footnote{Our code is available at:
\url{https://github.com/Dovermore/AuditAnalysis}}  This exploits the Markovian
nature of the sampling procedure and the low dimensionality of the (univariate)
statistics.  This approach allowed us to calculate---for elections with up to
tens of thousands of votes---exact values of each of the performance measures,
including the tail probabilities of the sampling distributions, which require
large sample sizes to estimate accurately by Monte Carlo.  We expect that with
some further optimisations our approach would be computationally feasible for
larger elections (up to 1 million votes).  The complexity largely depends on
the maximum sample size, $m$.  As long as this is moderate (thousands) our
approach is feasible.  For more complex audits (beyond two-candidate contests),
a Monte Carlo approach is likely more practical.


\section{Results}
\label{sec:results}

\subsection{Benchmarking results}

\subsubsection{Sample size distributions.}
\label{sec:boundary_pdf_cdf}

Different methods have different distributions of sample sizes;
\autoref{fig:pmf-distribution} shows these for a few methods when $p_T = 0.5$.
Some methods tend to stop early; others take many more samples.  Requiring a
minimum sample size might improve performance of some of the methods; see
\autoref{sec:exploring-improvements}.

\begin{figure}[ht]
\centering
\includegraphics[width=0.94\textwidth]{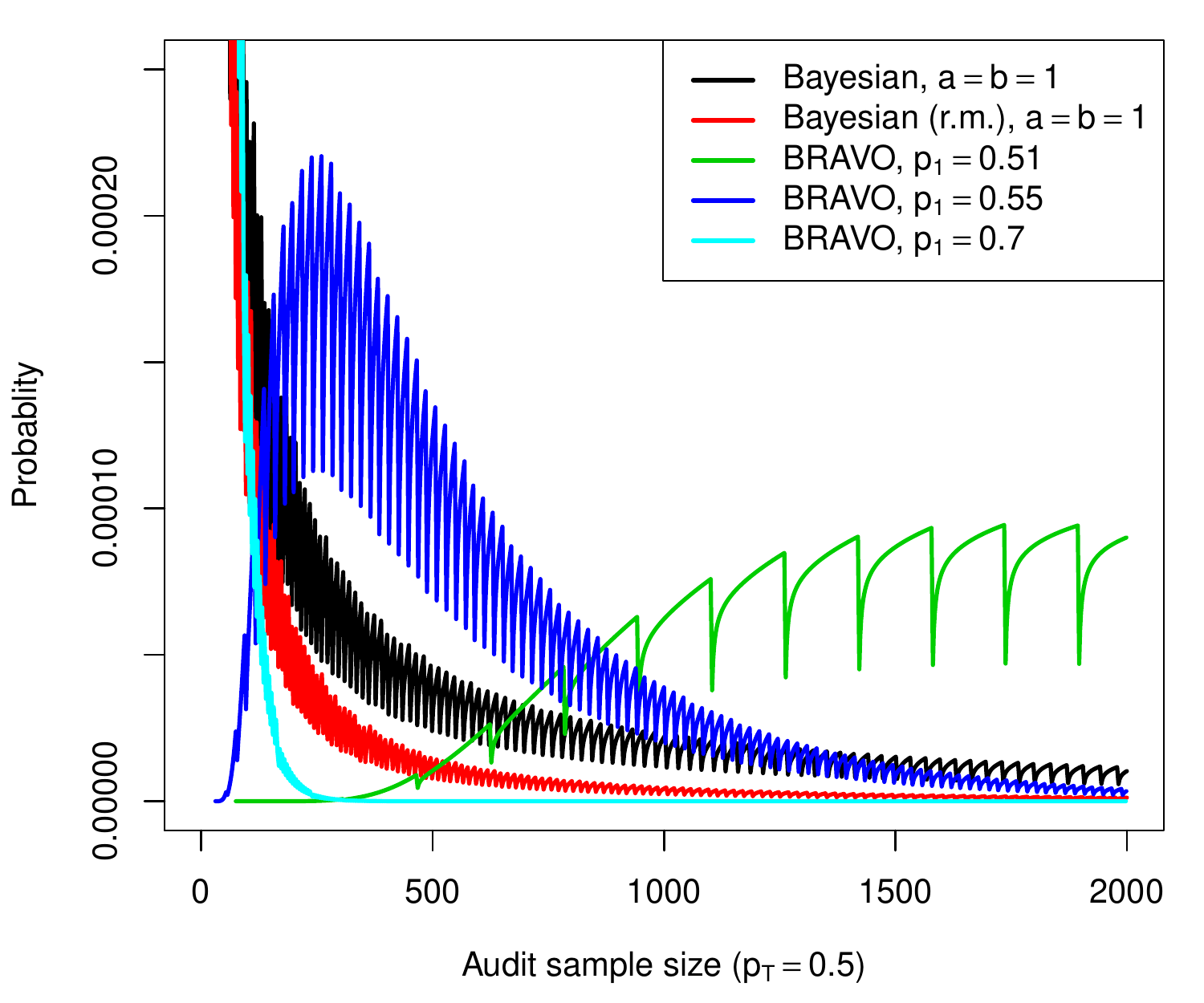}
\caption{\textbf{Sample size distributions.}  Audits of elections with $N =
20,000$ ballots, maximum sample size $m = 2,000$, and true vote share a tie
($p_T = 0.5$).  Each method is calibrated to have maximum risk 5\%.  The
depicted probabilities all sum to 0.05; the remaining 0.95 probability in each
case is on the event that the audit reaches the full sample size ($n = m$) and
progresses to a full manual tabulation.  `Bayesian (r.m.)' refers to the
Bayesian audit with a risk-maximizing prior.  The sawtooth pattern is due to
the discreteness of the statistics.}
\label{fig:pmf-distribution}
\end{figure}

\subsubsection{Mean sample sizes.}

We focus on average sample sizes as a measure of audit efficiency.
\autoref{tab:main-comparisons} shows the results of experiments with $N =
20,000$ and $m = 2,000$.  We discuss other experiments and performance measures
below.

\begin{table}
\centering

\caption{\textbf{Results from benchmarking experiments.}  Audits of elections
with $N = 20,000$ ballots and a maximum sample size $m = 2,000$.  The numeric
column headings refer to the value of $p_T$; the corresponding margin of
victory (\textsc{mov}) is also reported.  Each row refers to a specific
auditing method.  For calibrated methods, we report the threshold obtained.
For easier comparison, we present these on the nominal risk scale for BRAVO,
MaxBRAVO and ClipAudit (e.g.\ $\alpha = 1 / h$ for BRAVO), and on the upset
probability scale for the Bayesian methods ($\upsilon = 1 / (h + 1)$).  For the
experiments without calibration, we report the maximum risk of each method when
set to a `nominal' risk limit of 5\%.  We only report uncalibrated results for
methods that are automatically risk-limiting, as well as ClipAudit using its
`best fit' formula to set the threshold.  `Bayesian (r.m.)' refers to the
Bayesian audit with a risk-maximizing prior.  The numbers in bold are those
that are (nearly) best for the given experiment and choice of $p_T$.  The
section labelled `$n \geqslant 300$' refers to experiments that required the
audit to draw at least 300 ballots.}
\label{tab:main-comparisons}
\smallskip
\input{tables/benchmarking-experiments}
\end{table}

No method was uniformly best.  Given the equivalence of BRAVO and Bayesian
audits, the comparisons amount to examining dependence on the prior.

In general, methods that place more weight on close elections, such as BRAVO
with $p_1 = 0.55$ or a Bayesian audit with a moderately constrained prior ($a =
b = 100$) were optimal when $p_T$ was closer to 0.5.  Methods with substantial
prior weight on wider margins, such as BRAVO with $p_1 = 0.7$ and Bayesian
audits with the risk-maximizing prior, perform poorly for close elections.

Consistent with theory, BRAVO was optimal when the assumptions matched the
truth ($p_1 = p_T$).  However, our experiments violate the theoretical
assumptions because we imposed a maximum sample size, $m$.  (Indeed, when $p_1
= p_T = 0.51$, BRAVO is no longer optimal in our experiments.)

Two methods were consistently poor: BRAVO with $p_1 = 0.51$ and a Bayesian
audit with $a = b = 500$.  Both place substantial weight on a very close
election.

MaxBRAVO and ClipAudit, the two methods without a direct match to Bayesian
audits, performed similarly to a Bayesian audit with a uniform prior ($a = b =
1$).  All three are `broadly' tuned: they perform reasonably well in most
scenarios, even when they are not the best.

\subsubsection{Effect of calibration on the uncalibrated methods.}

For most of the automatically calibrated methods, calibration had only a small
effect on performance.  BRAVO with $p_1 = 0.51$ is an exception: it was very
conservative because it normally requires more than $m$ samples.

\subsubsection{Other election sizes and performance measures.}

The broad conclusions are the same for a range of values of $m$ and $N$, and
when performance is measured by quantiles of sample size or probability of
stopping without a full hand count rather than by average sample size.

\subsubsection{Sampling with vs without replacement.}

There are two ways to change our experiments to explore sampling with
replacement: (i)~construct versions of the statistics specifically  for
sampling with replacement; (ii)~leave the methods alone but sample with
replacement.  We explored both options, separately and combined; differences
were minor when $m \ll N$.

\subsection{Choosing between methods}

Consider the following two methods, which were the most efficient for different
election margins: (i)~BRAVO with $p_1 = 0.55$; (ii)~ClipAudit.  For $p_T =
0.52$, the mean sample sizes are 1,549 vs 1,630 (BRAVO saved 81 draws on
average).  For $p_T = 0.7$, the equivalent numbers are 85 vs 45 (ClipAudit
saved 40 draws on average).

Picking a method requires trade-offs involving resources, workload
predictability, and jurisdictional idiosyncrasies in ballot handling and
storage---as well as the unknown true margin.  Differences in expected sample
size across ballot-polling methods might be immaterial in practice compared to
other desiderata.

\subsection{Exploring changes to the methods}
\label{sec:exploring-improvements}

\subsubsection{Increasing the sampling increment (`round size').}

Increasing the number of ballots sampled in each `round' increases the chance
that the audit will stop without a full hand count but increases mean sample
size.  This is as expected; the limiting version is a single fixed sample of
size $n = m$, which has the highest power but loses the efficiency that early
stopping can provide. 

Increasing the sampling increment had the most impact on methods that tend to
stop early, such as Bayesian audits with $a = b = 1$, and less on methods that
do not, such as BRAVO with $p_1 = 0.51$.  Increasing the increment also
decreases the differences among the methods.  This makes sense because when the
sample size is $m$, the methods are identical (since all are calibrated to
attain the risk limit).

Considering the trade-off discussed in the previous section, since increasing
the sampling increment improves power but increases mean sample size, it
reduces effort when the election is close, but increases it when the margin is
wide.

\subsubsection{Increasing the maximum sample size ($m$).}

Increasing $m$ has the same effect as increasing the sampling increment: higher
power at the expense of more work on average.  This effect is stronger for
closer elections, since sampling will likely stop earlier when the margin is
wide.

\subsubsection{Requiring/encouraging more samples.}

The Bayesian audit with $a = b = 1$ tends to stop too early, so we tried two
potential improvements, shown in \autoref{tab:main-comparisons}.

The first was to impose a minimum sample size, in this case $n \geqslant 300$.
This is very costly if the margin is wide, since we would not normally require
this many samples.  However, it boosts the power of this method and reduces its
expected sample size for close contests.

A gentler way to achieve the same aim is to make the prior more informative, by
increasing $a$ and $b$.  When $a = b = 100$, we obtain largely the same benefit
for close elections with a much milder penalty when the margin is wide.  The
overall performance profile becomes closer to BRAVO with $p_1 = 0.55$.


\section{Discussion}
\label{sec:discussion}

We compared several ballot-polling methods both analytically and numerically,
to elucidate the relationships among the methods.  We focused on two-candidate
contests, which are building blocks for auditing more complex elections.  We
explored modifications and extensions to existing procedures.  Our benchmarking
experiments calibrated the methods to attain the same maximum risk.

Many `non-Bayesian' auditing methods are special cases of a Bayesian procedure
for a suitable prior, and Bayesian methods can be calibrated to be
risk-limiting (at least, in the two-candidate, all-valid-vote context
investigated here).  Differences among such methods amount to technical
details, such as choices of tuning parameters, rather than something more
fundamental.  Of course, upset probability \emph{is} fundamentally different
from risk.

No method is uniformly best, and most can be `tuned' to improve performance for
elections with either closer or wider margins---but not both simultaneously.
If the tuning is not extreme, performance will be reasonably good for a wide
range of true margins.  In summary:
\begin{enumerate}
\item If the true margin is known approximately, BRAVO is best.

\item Absent reliable information on the margin, ClipAudit and Bayesian audits
    with a uniform prior (calibrated to attain the risk limit) are efficient.

\item Extreme settings, such as $p_1 \approx 0.5$ or an overly informative
    prior may result in poor performance even when the margin is small.  More
    moderate settings give reasonable or superior performance if the maximum
    sample size is small compared to the number of ballots cast.
\end{enumerate}
Choosing a method often involves a trade-off in performance between narrow and
wide margins.

There is more to auditing than the choice of statistical inference method.
Differences in performance across many `reasonable' methods are small compared
to other factors, such as how ballots are organized and stored.

\textbf{Future work:} While we tried to be comprehensive in examining
ballot-polling methods for two-candidate contests with no invalid votes, there
are many ways to extend the analysis to cover more realistic scenarios.  Some
ideas include: (i)~more than two candidates and non-plurality social choice
functions; (ii)~invalid votes; (iii)~larger elections; (iv)~stratified samples;
(v)~batch-level audits; (vi)~multi-page ballots.

\bibliographystyle{splncs04}
\bibliography{evoteid2020-twocandidate}


\appendix

\section{Risk-limiting Bayesian audits with arbitrary priors}
\label{app:bayes-rla}

All of the results in this appendix are for \emph{ballot-polling audits of
two-candidate contests with no invalid votes}, the same as for the main text of
this paper.  For brevity, we often omit stating this assumption in the
mathematical statements below.

For such contests, Vora \cite{Vora2019Risk-LimitingElections} provided a
construction of a risk-limiting Bayesian audit, by taking a Bayesian audit with
an arbitrary prior and constructing a new prior from it that has the property
that a threshold on the upset probability is also a risk limit.

We extend that result to show that \emph{any} prior has a bounded maximum risk
and can therefore be used to conduct a risk-limiting audit.  Such a usage would
involve calculating a threshold on the upset probability that gives a
particular specified bound on the risk.  Our derivation also applies to a more
general class of Bayesian audits than considered by Vora
\cite{Vora2019Risk-LimitingElections}.

First, we need to define some more general procedures and further notation.

\subsection{General SPRT and Bayesian audits}

The SPRT-based and Bayesian audits presented in the main text only terminate
sampling in order to certify the election.  More generally, we can specify a
threshold of evidence for terminating the sampling and proceeding immediately
to a full manual tabulation (similar to reaching the sample size limit, $m$,
but the threshold could be reached much earlier).

The SPRT-based audit, as per Wald's original definition, is as follows:
\[
\begin{cases}
\text{if } S_n > \frac{1 - \beta}{\alpha},
    & \text{terminate and certify,} \\
\text{if } S_n < \frac{\beta}{1 - \alpha},
    & \text{proceed to full tabulation,} \\
\text{otherwise,}
    & \text{continue sampling.} \\
\end{cases}
\]

The statistic $S_n$ is the likelihood ratio, the same as our earlier definition
from \autoref{sec:sprt}.  The only new component here is $\beta$.  When $\beta
= 0$, we recover the earlier definition, noting that the second inequality here
will never be satisfied because $S_n$ cannot be negative.  Like the earlier
definition, this version of the SPRT has a risk limit of $\alpha$.\footnote{As
a hypothesis test (of two simple hypotheses; see \autoref{sec:sprt}), it also
has the property that the power is at least $1 - \beta$.  In statistical
parlance, $\alpha$ and $\beta$ are limits on the type I and type II error
respectively.}

Analogously, we define a more general Bayesian audit as follows:
\[
\begin{cases}
\text{if } S_n > \frac{1 - \upsilon}{\upsilon},
    & \text{terminate and certify,} \\
\text{if } S_n < \frac{1 - \phi}{\phi},
    & \text{proceed to full tabulation,} \\
\text{otherwise,}
    & \text{continue sampling.} \\
\end{cases}
\]
The statistic $S_n$ is the Bayes factor, the same as our earlier definition
from \autoref{sec:bayesian}.  If the prior gives equal weight to the two
hypotheses, $\Pr(H_0) = \Pr(H_1)$, then the inequalities above are equivalent
to more straightforward ones in terms of the upset probability, as follows:
\[
\begin{cases}
\text{if } \Pr(H_0 \mid Y_n) < \upsilon,
    & \text{terminate and certify,} \\
\text{if } \Pr(H_0 \mid Y_n) > \phi,
    & \text{proceed to full tabulation,} \\
\text{otherwise,}
    & \text{continue sampling.} \\
\end{cases}
\]
Setting $\phi = 1$ recovers our earlier definition of the Bayesian audit.

\subsection{Correspondence between the SPRT and Bayesian audits}

In \autoref{sec:sprt}, we noted that the SPRT is equivalent to a Bayesian audit
with a particular choice of prior (equal point masses on $p_0$ and $p_1$) and a
specific choice of threshold.  We can extend this correspondence to the more
general definitions above.  We keep the choice of prior the same (equal point
masses)---this makes $S_n$ equivalent under both audit methods---and equate the
two sets of thresholds,
\[
    \frac{1 - \beta}{\alpha} = \frac{1 - \upsilon}{\upsilon}
    \quad \text{and} \quad
    \frac{\beta}{1 - \alpha} = \frac{1 - \phi}{\phi}.
\]
We can solve these for either $\alpha$ and $\beta$, or $\upsilon$ and $\phi$,
which respectively gives:
\begin{align*}
    \alpha   &= \frac{\upsilon (2\phi - 1)}{\phi - \upsilon}      &
    \upsilon &= \frac{\alpha}{1 - \beta + \alpha}                 \\
    \beta    &= \frac{(1 - \phi)(1 - 2\upsilon)}{\phi - \upsilon} &
    \phi     &= \frac{1 - \alpha}{1 - \beta + \alpha}.
\end{align*}
In other words, the SPRT is equivalent to a Bayesian audit with a prior having
equal point masses on $p_0$ and $p_1$, and the thresholds $\upsilon$ and $\phi$
set to the values given above (in terms of the desired $\alpha$ and $\beta$).

As a corollary, a Bayesian audit with such a prior will have a risk limit given
by $\upsilon (2\phi - 1) / (\phi - \upsilon)$.

\subsubsection{Special cases.}

The results in the main text were for the special case where $\beta = 0$ and
$\phi = 1$.  Under this case, the above correspondence with the SPRT simplifies
to $\alpha = \upsilon / (1 - \upsilon)$ and $\upsilon = \alpha / (1 + \alpha)$.

Vora \cite{Vora2019Risk-LimitingElections} considered a more restricted version
of the Bayesian audit where $\phi = 1 - \upsilon$ (symmetric thresholds on the
upset probability).  This is a special case of the general definition given
here.  Under this special case, the above correspondence with the SPRT
simplifies to $\alpha = \beta = \upsilon$ (the risk limit is the same the the
threshold on the upset probability).

The relative sizes of the risk limit ($\alpha$) and the upset probability
threshold ($\upsilon$) are sometimes of interest.  It is instructive to look at
the correspondence with the SPRT and consider different cases.  In the special
case in the main text, we have $\alpha = \upsilon / (1 - \upsilon) > \upsilon$;
the risk limit is larger than the threshold.  For Vora's special case, the two
values coincide.  We can also construct audits where the risk limit is
\emph{smaller} than the threshold by setting $\phi < 1 - \upsilon$; such audits
are more likely to terminate early and proceed to a full tabulation, hence
reducing the risk.

\subsubsection{The risk-maximizing prior.}

Vora \cite{Vora2019Risk-LimitingElections} defined a `risk-maximising prior' in
order to construct a Bayesian audit for which the threshold on the upset
probability was also a risk limit.  From the more general analysis above, we
see that this construction relies on symmetric thresholds ($\phi = 1 -
\upsilon$) and will not work in general.  For example, in the special case
considered in the main text here, to obtain a risk limit of $\alpha$ we need to
set a stricter threshold on the upset probability: $\upsilon = \alpha / (1 +
\alpha) < \alpha$.  In terms of the Bayes factor, this translates to: $S_n > 1
/ \alpha$.

\subsection{Further definitions and notation}

An audit will result in a sequence of $n$ sampled ballots.  The sequence ends
either once the audit termination condition is met resulting in the
certification of the election, or otherwise once it has progressed to a full
manual tabulation of the votes (e.g.\ upon reaching $n = m$ without
certification).  Let the $\Lambda$ be the set of all sequences that lead to
certification, and $\bar\Lambda$ the set of all those that do not.

When the reported winner is \emph{not} the true winner, the sequences in
$\Lambda$ will result in miscertification while those in $\bar\Lambda$ will
lead to discovering the true winner.  The risk of the audit will be the
probability of obtaining a sequence from $\Lambda$ (rather than $\bar\Lambda$).
In other words,
\begin{equation}
\label{eqn:certprob}
  \Pr(\text{certify} \parallel T, N)
= \sum_{s \in \Lambda} \Pr(s \parallel T, N),
\end{equation}
where $s = (X_1, X_2, \dots, X_n)$ here represents an arbitrary sequence that
leads to certification, and as before $T$ is the true total number of votes for
the winner and $N$ the total number of cast ballots.  Note that $n$ is not
fixed; different sequences can terminate at different values of $n$.  The
summand is the probability of observing a specific sequence.  If sampling with
replacement, this will be a product of Bernoulli probabilities,
\begin{equation}
\label{eqn:seqprob-with}
  \Pr(s \parallel T, N) = \Pr(s \parallel p_T) = p_T^{Y_n} (1 - p_T)^{n - Y_n}.
\end{equation}
Note that this is the same a binomial probability but without the binomial
coefficient, $\binom{n}{Y_n}$.  It represent the probability of a given
\emph{ordered} sequence of ballots.  If sampling without replacement, the
probability of the sequence will be an ordered version of a hypergeometric
probability,
\begin{equation}
\label{eqn:seqprob-without}
  \Pr(s \parallel T, N)
= \frac{1}{\binom{n}{Y_n}}
  \frac{\binom{T}{Y_n}\binom{N - T}{n - Y_n}}{\binom{N}{n}}
= \frac{T^{(Y_n)} (N - T)^{(n - Y_n)}}{N^{(n)}},
\end{equation}
using the notation $n^{(m)} = n (n - 1) \cdots (n - m + 1) = n! / (n - m)!$.

Before we state the main results of this appendix, we define a property of the
auditing methods that we need for later proofs.
\begin{definition}
\label{dfn:sensible}
An auditing method is called \emph{sample-coherent} if for every sequence that
results in certification ($s \in \Lambda$), the sampled final tally will be in
favour of the reported winner, i.e.\ $Y_n / n > 0.5$.
\end{definition}
This property would be expected of any sensible auditing method, but is not
mathematically guaranteed.  Indeed, it is possible to define stopping rules for
an audit such that this property does not hold.  Here, we only consider
sample-coherent auditing methods.\footnote{We are pretty sure this covers all
of the methods defined in this paper, but have not mathematically verified it
for each one.}

\subsection{Bayesian audits are risk-limiting}

\begin{lemma}
\label{lemma:maxrisk}
The maximum risk of a sample-coherent ballot-polling audit of a plurality
contest with two candidates and no invalid votes is given by the
(mis)certification probability when the true tally gives equal votes for each
candidate ($T = \frac{N}{2}$, $p_T = \frac{1}{2}$), or the closest possible
such non-winning tally in the case where we have odd number of cast ballots ($T
= \frac{N - 1}{2}$).
\end{lemma}
\begin{proof}
We prove this lemma by showing the summand in \eqref{eqn:certprob} is
monotonically increasing in $T$ when $H_0$ is true.  This is the same technique
used to prove Theorem 2 in Vora~\cite{Vora2019Risk-LimitingElections}, for the
case of $N$ odd and sampling without replacement.  Here we extend the argument
to also cover $N$ even and sampling with replacement.

For sampling with replacement, the summand is given by
\eqref{eqn:seqprob-with}.  As a function of $p_T$ over the unit interval, this
is strictly increasing until it reaches a maximum at $p_T = Y_n / n$ and
strictly decreasing thereafter.  This is easily shown via calculus (note also
that \eqref{eqn:seqprob-with} is a binomial likelihood with $Y_n / n$ as the
corresponding maximum likelihood estimate).  By \autoref{dfn:sensible}, this
maximum occurs at a value $p_T > 0.5$.  Therefore, when $H_0$ is true ($p_T
\leqslant 0.5$), the summand is monotonically increasing in $p_T$, and
therefore also in $T$.

For sampling without replacement, the summand is given by
\eqref{eqn:seqprob-without}.  Consider this as a function of $T$.  We will
increment this by 1 (replacing $T$ with $T + 1$), take the ratio with the
original version, and show this ratio is positive.  Let the ratio be $G$.  We
have
\[
G = \frac{\Pr(s \parallel T + 1, N)}
         {\Pr(s \parallel T    , N)}
  = \frac{T + 1}{T - Y_n + 1} \times
    \frac{(N - T) - (n - Y_n)}{N - T}.
\]
We assume that $H_0$ is still true when $T$ is incremented, that is we have $(T
+ 1) / N \leqslant 0.5$.  Also, by \autoref{dfn:sensible} we have $0.5 < Y_n /
n$.  Therefore,
\[
    \frac{T + 1}{N} < \frac{Y_n}{n}.
\]
This inequality still holds if we convert both of these probabilities into
odds,
\begin{align*}
    \frac{(    T + 1) / N}
         {(N - T - 1) / N}
 &< \frac{     Y_n  / n}
         {(n - Y_n) / n}  \\
    \frac{    T + 1}
         {N - T - 1}
 &< \frac{    Y_n}
         {n - Y_n}.
\end{align*}
We also trivially have
\[
    \frac{T + 1}{N - T} < \frac{T + 1}{N - T - 1}
\]
as long as $N - T > 1$, i.e.\ as long as one vote has been cast in favour of
the losing candidate.  This will be true for all practical scenarios of
interest since $N$ will be large (certainly, much larger than 2), and under
$H_0$ we must have $T$ be at most $N / 2$.

Combining the previous two inequalities gives,
\begin{align*}
    \frac{T + 1}{N - T}
 &< \frac{Y_n}{n - Y_n}  \\
    \frac{n - Y_n}{N - T}
 &< \frac{Y_n}{T + 1}  \\
    1 - \frac{Y_n}{T + 1}
 &< 1 - \frac{n - Y_n}{N - T}  \\
    \frac{T + 1 - Y_n}{T + 1}
 &< \frac{(N - T) - (n - Y_n)}{N - T}  \\
    1
 &< \frac{T + 1}{T + 1 - Y_n} \times
    \frac{(N - T) - (n - Y_n)}{N - T}
  = G.
\end{align*}
In other words, the ratio is positive.  Therefore, the summand given by
\eqref{eqn:seqprob-without} is monotonically increasing in $T$ as long as $H_0$
is true.
\qed
\end{proof}

\begin{lemma}
\label{lemma:bayesmonotone}
The risk of a Bayesian audit is a monotone increasing function of $\upsilon$,
the threshold on the upset probability.  In other words, relaxing the
threshold leads to greater risk.
\end{lemma}
\begin{proof}
If $\upsilon$ is increased (i.e.\ less stringent evidence threshold), then:
\begin{itemize}
\item Any sequence in $\Lambda$ remains in $\Lambda$, since if it passed the
    earlier, stricter threshold then it will also pass the newer, looser one.
    The only change will be that the audit possibly terminates earlier (i.e.\
    the sample size at which the audit stops is reduced).
\item Some sequences in $\bar\Lambda$ move to $\Lambda$, because they now meet
    the newer, relaxed threshold.
\end{itemize}
Therefore, overall there will be a shift of probability from $\bar\Lambda$ to
$\Lambda$.  Therefore the risk has increased.  Note that this is true for any
value of $T$, including for the value that maximises the risk.
\qed
\end{proof}
Note that \autoref{lemma:bayesmonotone} also holds for any audit procedure with
a stopping rule of the form, ``stop if $f(\text{sample data}) \geqslant
\upsilon$''.  A Bayesian audit is just one such procedure.

\begin{corollary}
For any prior, Bayesian ballot-polling audits of a two-candidate, no invalid
vote plurality contest can be calibrated to be risk-limiting, in the sense that
for any risk limit $\alpha \in (0,1)$, there exists $\upsilon \in (0, 1)$ such
that terminating the audit when the upset probability is not greater than
$\upsilon$ limits the risk of the procedure to $\alpha$.  (Typically, this
requires $\upsilon < \alpha$.)
\end{corollary}
\begin{proof}
\autoref{lemma:maxrisk} shows that any Bayesian audit has a bounded maximum
risk.  The monotonic relationship from \autoref{lemma:bayesmonotone} implies
that we can reduce this risk by imposing a stricter threshold on the upset
probability.  In particular, we can reduce it until the maximum risk is less
than any pre-specified limit.  Thus, we can use any Bayesian audit in a
risk-limiting fashion.
\qed
\end{proof}
To implement this in practice we need to be able calculate the maximum risk for
any given threshold and optimise the threshold value until this risk is smaller
than the specified limit.  This is straightforward for the two-candidate case
via either simulation or exact calculation, since we know the value of $T$ that
gives rise to the maximum risk.  Such a calculation would need to be done
separately for any given choice of sampling scheme and prior.


\section{KMart as a Bayesian audit}
\label{app:kmart-bayes}

This appendix shows a proof that for sampling with replacement, KMart is
equivalent to a Bayesian audit with a risk-maximising prior for the reported
winner's true vote tally, uniform on $(\frac{1}{2}, 1]$.  It also introduces a
more general version of the test statistic that corresponds to an arbitrary
risk-maximising prior.  Both results are shown for a simple two-candidate
contest.

\subsection{KMart is equivalent to a Bayesian audit}

As in \autoref{sec:other-methods}, we let $t = \E(X_1 \parallel H_0)$, the true
proportion of votes for the reported winner under the null hypothesis.  Since
we assume sampling without replacement, $\{X_i\}$ is a sequence of Bernoulli
trials with success probability $t$ if the null is true.  As explained
previously, we set $t = 0.5$ since that is the incorrect outcome for which the
risk of the method is largest.

In practice we will always have a finite number of total votes, and thus a
realistic model would have the support of $t$ be a discrete set (i.e.\ values
of the form $k / N$, where $k$ is the total number of votes in favour of the
reported winner).  However, for mathematical convenience here we will allow the
support of $t$ to be the unit interval, which is continuous.

To distinguish the test statistics ($S_n$) from KMart and Bayesian audits, we
will refer to the KMart statistic by $A_n$ and the Bayesian statistic by
$B_n$.

\subsubsection{KMart audits.}

For sampling with replacement, the KMart statistic is:
\[
A_n = \int_0^1 \prod_{i=1}^n \left(\gamma \left[\frac{X_i}{t} - 1\right]
    + 1\right) d\gamma.
\]
Since we are working with $t = \frac{1}{2}$, we can rewrite this expression,
\[
A_n = 2^n \int_0^1 \prod_{i=1}^n \left(\gamma \left[X_i - \frac{1}{2}\right]
    + \frac{1}{2}\right) d\gamma.
\]
For a specified risk limit, $\alpha$, the audit proceeds until $A_n
> 1 / \alpha$, at which point the election is certified ($H_0$ is rejected), or
is otherwise terminated in favour of doing a full manual tabulation.

\subsubsection{Bayesian audits.}

The Bayesian statistic is the BF:
\[
B_n = \frac{\Pr(X_1,\dots,X_n \mid H_1)}{\Pr(X_1,\dots,X_n \mid H_0)}.
\]
For sampling with replacement, the likelihood is the product of Bernoulli mass
functions and can be written completely in terms of $Y_n$, the total number of
sampled votes in favour of the winner,
\[
    \Pr(X_1,\dots,X_n \mid p) = p^{Y_n} (1 - p)^{n - Y_n}.
\]
We limit our discussion to risk-maximising prior distributions.  Reminder:
these place a probability mass of $\frac{1}{2}$ on the value of $p =
\frac{1}{2}$, and the remaining probability is over the set $p \in
(\frac{1}{2}, 1]$.  

The denominator of the BF is simple: the likelihood of the sample at the
(point) null value,
\[
  \Pr(X_1,\dots,X_n \mid H_0)
= \Pr\left(X_1,\dots,X_n \mid p = \tfrac{1}{2}\right)
= \left(\tfrac{1}{2}\right)^{Y_n} \left(\tfrac{1}{2}\right)^{n - Y_n}
= 1 / 2^n.
\]
The numerator requires integrating over the prior under $H_1$,
\[
  \Pr(X_1,\dots,X_n \mid H_1)
= \int_{\frac{1}{2}}^1 p^{Y_n} \left(1 - p\right)^{n - Y_n} f(p) \, dp.
\]
Putting these together gives,
\[
B_n = 2^n \int_{\frac{1}{2}}^1 p^{Y_n} \left(1 - p\right)^{n - Y_n} f(p) \, dp.
\]
Similar to KMart, a Bayesian audit proceeds until $B_n > 1 / \alpha$.

\subsubsection{Equivalence.}

Both $A_n$ and $B_n$ are expressed as integrals but with the $X_i$ in different
`places' in the integrand.  The key to showing they are equivalent is to notice
that the $X_i$ are binary variables, which allows us to set up an identity that
relates the two ways of writing the integral.  Specifically, we have the
following identity,
\[
  \gamma \left(X_i - \frac{1}{2}\right) + \frac{1}{2}
= \left(\frac{1 + \gamma}{2}\right)^{X_i}
  \left(\frac{1 - \gamma}{2}\right)^{1 - X_i}.
\]
This allows us to rewrite $A_n$,
\[
A_n
= 2^n \int_0^1 \left(\frac{1 + \gamma}{2}\right)^{Y_n}
    \left(\frac{1 - \gamma}{2}\right)^{n - Y_n} d\gamma
= \int_0^1 \left(1 + \gamma\right)^{Y_n}
    \left(1 - \gamma\right)^{n - Y_n} d\gamma.
\]
Next, let $\gamma = 2p - 1$ and change the variable of integration,
\[
A_n
= \int_{\frac{1}{2}}^1 (2p)^{Y_n} (2 - 2p)^{n - Y_n} 2 \, dp
= 2^n \int_{\frac{1}{2}}^1 p^{Y_n} \left(1 - p\right)^{n - Y_n} 2 \, dp.
\]
Finally, note that this is identical to $B_n$ if we set the prior to be uniform
over $H_1$, i.e.\ $f(p) = 2$.

In other words, a KMart audit is equivalent to a Bayesian audit that uses a
risk-maximising prior that is uniform on $(\frac{1}{2}, 1]$.

\subsection{Extending KMart to arbitrary priors}

From the above result, we can see that $\gamma$ plays a similar role to $p$.
The somewhat arbitrary integral over $\gamma$ used to define $A_n$ can be
generalised by specifying a weighting function $g(\gamma)$,
$$
A_n = \int_0^1 \prod_{i=1}^n \left(\gamma \left[\frac{X_i}{t} - 1\right]
    + 1\right) g(\gamma) \, d\gamma.
$$
Applying the same transformations as above gives
$$
A_n = 2^n \int_{\frac{1}{2}}^1 p^{Y_n} \left(1 - p\right)^{n - Y_n}
    2 \times g(2p - 1) \, dp.
$$
In other words, this generalised version of KMart is equivalent to a Bayesian
audit with the following risk-maximising prior:
$$
f(p) = 2 \times g(2p - 1).
$$
The original KMart is the special case where $g(\cdot) = 1$.

\subsection{Efficient computation by exploiting the equivalence}

We can use the above equivalence to develop fast ways to compute the KMart
statistic (when sampling with replacement), by relating it to standard Bayesian
calculations using conjugate priors.

First, we show that if we take a conjugate prior distribution, truncate it, and
add some point masses, the resulting distribution is still conjugate.  Then we
use this result to write a formula for the posterior distribution for the same
case as above (simple two-candidate election, sampling with replacement).

\subsubsection{Truncation and point masses preserve conjugacy.}

(The proofs shown here are not too hard to derive and may well be described
elsewhere.)

Suppose we have a single parameter, $\theta$, some data, $D$, a likelihood
function, $L(\theta \mid D)$, and a conjugate prior distribution, $f(\theta)$.
That means we have,
$$
f(\theta \mid D) \propto L(\theta \mid D) f(\theta).
$$
Let the normalising constant be,
$$
k = \int L(\theta \mid D) f(\theta) d\theta.
$$
This allows us to express the posterior as,
$$
f(\theta \mid D) = \frac{1}{k} L(\theta \mid D) f(\theta),
$$
The sections that follow each start with these definitions and transform the
prior in various ways.

\paragraph{Truncation.}

Truncate the prior to a subset $S$ (i.e.\ we only allow $\theta \in S$).  Write
this truncated prior as,
$$
f^*(\theta) = f(\theta) \frac{I_S(\theta)}{z_S},
$$
where $I_S(\theta)$ is the indicator function that takes value 1 when $\theta
\in S$, and $z_S = \int f(\theta) I_S(\theta) d\theta = \int_S f(\theta)
d\theta$ is the normalising constant due to truncation.

If we use this prior, we get the posterior
$$
f^*(\theta \mid D) = \frac{1}{k^*} L(\theta \mid D) f^*(\theta),
$$
where,
$$
k^* = \int L(\theta \mid D) f^*(\theta) d\theta.
$$
Expanding this out gives,
$$
f^*(\theta \mid D) = \frac{1}{k^* z_S} L(\theta \mid D) f(\theta) I_S(\theta)
           = \frac{k}{k^* z_S} f(\theta \mid D) I_S(\theta).
$$
This is the original posterior truncated to $S$.  Thus, the truncation results
in staying within the same family of (truncated) probability distributions,
which means this family is conjugate.

\paragraph{Adding a point mass.}

Define a `spiked' prior where we add a point mass at $\theta_0$,
$$
f^*(\theta) = u \, \delta_{\theta_0}(\theta) + v f(\theta),
$$
where $u + v = 1$.  In other words, a mixture distribution with mixture weights
$u$ and $v$.  The normalising constant is,
$$
k^* = \int L(\theta \mid D) f^*(\theta) d\theta = u L(\theta_0 \mid D) + v k.
$$
We can write the posterior as,
$$
f^*(\theta \mid D) = \frac{1}{k^*} L(\theta \mid D) f^*(\theta)
           = \frac{u \, L(\theta_0 \mid D)}{k^*} \delta_{\theta_0}(\theta) +
             \frac{v k}{k^*} f(\theta \mid D).
$$
This is a `spiked' version of the original posterior.  You can see this more
clearly by defining,
$$
u^* = \frac{u \, L(\theta_0 \mid D)}{k^*}, \quad v^* = \frac{v k}{k^*},
$$
where $u^* + v^* = 1$.  Thus, `spiking' a distribution results in a conjugate
family.  Note that the mixture weights get updated as we go from the prior to
the posterior.

\paragraph{Truncating and adding point masses.}

We can combine both of the previous operations and we will still retain
conjugacy.  In fact, due to the generality of the proof, we can apply each one
an arbitrary number of times, e.g.\ to add many point masses.

\subsubsection{Application to KMart.}

When sampling with replacement, the conjugate prior for $p$ is a beta
distribution.

We showed earlier that KMart was equivalent to using a risk-maximising prior.
Starting with any beta distribution, we can form the corresponding
risk-maximising prior by truncating to $p \in (\frac{1}{2}, 1]$ and adding a
probability mass of $\frac{1}{2}$ at $p = \frac{1}{2}$.  Based on the argument
presented above, this prior is conjugate.  Moreover, we can express the
posterior in closed form.

Let the original prior be $p \sim \text{Beta}(a, b)$.  The risk-maximising
prior retains the functional form of this prior for $p > \frac{1}{2}$ and also
has a mass of $\frac{1}{2}$ at $p = \frac{1}{2}$.

After we observe a sample of size $n$ from the audit, we have a posterior with
an updated probability mass at $p = \frac{1}{2}$.  This mass will be the upset
probability.  We can derive an expression for it using equations similar to
above (it will correspond to $u^*$ using the notation from above).

Let $f(p)$ be the pdf of the original beta prior, $F(p)$ be its cdf, $S =
(\frac{1}{2}, 1]$ the truncation region, $F'(p)$ the cdf of the
beta-distributed portion of the posterior (i.e.\ the posterior distribution if
we use the original beta prior), and $B(\cdot, \cdot)$ be the beta function.
We have,
$$
k^* = \frac{1}{2} \left(\frac{1}{2}\right)^n + \frac{1}{2} \frac{k'}{z_S},
$$
where
$$
z_S = \int_\frac{1}{2}^1 f(p) dp = 1 - F\left(\frac{1}{2}\right)
$$
and
$$
k' = \int_\frac{1}{2}^1 L(p \mid D) f(p) dp
   = \frac{B(Y_n + a, n - Y_n + b)}{B(a, b)}
     \left(1 - F'\left(\frac{1}{2}\right)\right).
$$
Putting these together gives,
$$
k^* = \frac{1}{2^{n + 1}} +
      \frac{1}{2} \times
      \frac{B(Y_n + a, n - Y_n + b)}{B(a, b)} \times
      \frac{1 - F'\left(\frac{1}{2}\right)}
           {1 - F\left(\frac{1}{2}\right)}.
$$
The upset probability is,
$$
u^* = \frac{\frac{1}{2^{n + 1}}}{k^*}.
$$
These quantities will be straightforward to calculate as long we have efficient
ways to calculate:
\begin{enumerate}
\item The beta function.
\item The cdf of a beta distribution.
\end{enumerate}
Both have fast implementations in R.\footnote{\url{https://www.r-project.org/}}


\end{document}

%% file: tables/benchmarking-experiments.tex
\begin{tabular}{l@{\hskip 6pt}r@{\hskip 12pt}rrrlrrrrr}
\toprule
&  &  \multicolumn{3}{c}{\textbf{Power (\%)}} & &
      \multicolumn{5}{c}{\textbf{Mean sample size}}  \\
      \cmidrule{3-5}
      \cmidrule{7-11}
   &   $p_T$ (\%) $\rightarrow$ &
52 & 55 & 60 & \phantom{a} &
52 & 55 & 60 & 64 & 70  \\
\textbf{Method} &
\textsc{mov} (\%) $\rightarrow$ &
 4 & 10 & 20 &             &
 4 & 10 & 12 & 28 & 40 \\
\midrule
\textbf{Calibrated} & $\alpha$ or $\upsilon$ (\%)  \\
Bayesian, $a = b =   1$ & 0.2 & 35 & \textbf{99} & \textbf{100} &  & 1623 &
                                    637 & \textbf{172} &  \textbf{90} &  46 \\
Bayesian, $a = b = 100$ & 1.2 & 48 & \textbf{100} & \textbf{100} &  &
                                      \textbf{1551} & 616 & 232 & 150 &  97 \\
Bayesian, $a = b = 500$ & 3.6 & \textbf{53} & \textbf{100} & \textbf{100} &
                                             & 1582 & 709 & 318 & 219 & 149 \\
Bayesian (r.m.), $a = b = 1$ & 6.1 & 19 & 94 & \textbf{100} &  &
                                      1742 & 813 & 185 &  \textbf{89} &  41 \\
BRAVO, $p_1 = 0.7$ & 5.8 & 9 & 21 & 84 &  & 1828 & 1592 & 530 &  95 &
                                                                \textbf{37} \\
BRAVO, $p_1 = 0.55$ & 5.3 & 37 & \textbf{99} & \textbf{100} & & \textbf{1549}
                                 & \textbf{562} & 196 & 129 &  85 \\
BRAVO, $p_1 = 0.51$ & 22.7 & \textbf{55} & \textbf{100} & \textbf{100} &
                                             & 1617 & 791 & 384 & 272 & 190 \\
MaxBRAVO   & 1.6 & 30 & \textbf{98} & \textbf{100} & & 1660 & 680 & 177 &
                                                          \textbf{91} &  45 \\
ClipAudit  & 4.7 & 33 & \textbf{98} & \textbf{100} & & 1630 & 639 &
                                           \textbf{169} & \textbf{89} &  45 \\
\addlinespace
\textbf{Calibrated, $n \geqslant 300$}  & $\alpha$ or $\upsilon$ (\%)  \\
Bayesian, $a = b = 1$ & 0.6 & 45 & \textbf{99} & \textbf{100} &  & 1547 & 601
& \textbf{311} & \textbf{300} & \textbf{300} \\
Bayesian (r.m.), $a = b = 1$ & 34.4 & 39 & \textbf{99} & \textbf{100} &  &
1554 & 587 & \textbf{307} & \textbf{300} & \textbf{300} \\
BRAVO, $p_1 = 0.7$ & 100.0 & 0 & 6 & 83 &  & 1994 & 1900 & 708 & 309 &
\textbf{300} \\
BRAVO, $p_1 = 0.55$ & 6.0 & 38 & \textbf{99} & \textbf{100} &  &
\textbf{1545} & \textbf{583} & \textbf{309} & \textbf{300} & \textbf{300} \\
BRAVO, $p_1 = 0.51$ & 22.7 & \textbf{55} & \textbf{100} & \textbf{100} &  &
1617 & 791 & 392 & 313 & \textbf{300} \\
MaxBRAVO   & 5.0 & 44 & \textbf{99} & \textbf{100} &  & 1546 & 595 &
\textbf{310} & \textbf{300} & \textbf{300} \\
ClipAudit  & 11.4 & 44 & \textbf{99} & \textbf{100} &  & \textbf{1545} & 595 &
\textbf{310} & \textbf{300} & \textbf{300} \\
\addlinespace
\textbf{Uncalibrated}   & Risk (\%)  \\
Bayesian (r.m.), $a = b = 1$ & 3.7 & 17 & 93 & \textbf{100} &  & 1785 & 864 &
198 &  95 &  44 \\
BRAVO, $p_1 = 0.7$  & 4.3 & 8 & 20 & 83 &  & 1846 & 1621 & 552 &  99 &
\textbf{38} \\
BRAVO, $p_1 = 0.55$ & 4.7 & \textbf{37} & \textbf{98} & \textbf{100} &  &
\textbf{1561} & \textbf{572} & 200 & 131 &  86 \\
BRAVO, $p_1 = 0.51$ & 0.029 & 6 & 89 & \textbf{100} &  & 1985 & 1505 & 760 &
542 & 377 \\
ClipAudit   & 5.1 & \textbf{34} & \textbf{98} & \textbf{100} &  & 1618 & 628 &
\textbf{167} &  \textbf{88} &  45 \\
\bottomrule
\end{tabular}

%% file: evoteid2020-twocandidate.bbl
\begin{thebibliography}{10}
\providecommand{\url}[1]{\texttt{#1}}
\providecommand{\urlprefix}{URL }
\providecommand{\doi}[1]{https://doi.org/#1}

\bibitem{raire}
Blom, M., Stuckey, P.J., Teague, V.J.: Ballot-polling risk limiting audits for
  {IRV} elections. In: Electronic Voting. pp. 17--34. Springer, Cham (2018)

\bibitem{Kulldorff2011ASurveillance}
Kulldorff, M., Davis, R.L., Kolczak, M., Lewis, E., Lieu, T., Platt, R.: A
  maximized sequential probability ratio test for drug and vaccine safety
  surveillance. Sequential Analysis  \textbf{30}(1),  58--78 (2011).
  \doi{10.1080/07474946.2011.539924}

\bibitem{Lindeman2012BRAVO:Outcomes}
Lindeman, M., Stark, P.B., Yates, V.S.: {BRAVO}: Ballot-polling risk-limiting
  audits to verify outcomes. In: 2012 Electronic Voting Technology
  Workshop/Workshop on Trustworthy Elections (EVT/WOTE '12) (2012)

\bibitem{securing2018}
{National Academies of Sciences, Engineering, and Medicine}: Securing the Vote:
  Protecting {American} Democracy. The National Academies Press, Washington, DC
  (Sep 2018). \doi{10.17226/25120}

\bibitem{Rivest2017ClipAudit:Audit}
{Rivest}, R.L.: {C}lip{A}udit: A simple risk-limiting post-election audit.
  arXiv e-prints arXiv:1701.08312 (Jan 2017)

\bibitem{RivestAElections}
Rivest, R.L., Shen, E.: A {B}ayesian method for auditing elections. In: 2012
  Electronic Voting Technology/Workshop on Trustworthy Elections (EVT/WOTE '12)
  (2012)

\bibitem{stark08a}
Stark, P.: Conservative statistical post-election audits. Ann. Appl. Stat.
  \textbf{2},  550--581 (2008), \url{http://arxiv.org/abs/0807.4005}

\bibitem{shangrla}
Stark, P.: Sets of half-average nulls generate risk-limiting audits:
  {SHANGRLA}. Voting '20  \textbf{in press} (2020), preprint:
  \url{http://arxiv.org/abs/1911.10035}

\bibitem{stark2009}
Stark, P.B.: Risk-limiting postelection audits: Conservative {$P$}-values from
  common probability inequalities. IEEE Transactions on Information Forensics
  and Security  \textbf{4}(4),  1005--1014 (Dec 2009).
  \doi{10.1109/TIFS.2009.2034190}

\bibitem{StarkTeague2014}
Stark, P.B., Teague, V.: Verifiable {E}uropean elections: Risk-limiting audits
  for {D}'{H}ondt and its relatives. {USENIX} Journal of Election Technology
  and Systems ({JETS})  \textbf{1}(3),  18--39 (Dec 2014),
  \url{https://www.usenix.org/jets/issues/0301/stark}

\bibitem{Vora2019Risk-LimitingElections}
{Vora}, P.L.: {Risk-Limiting Bayesian Polling Audits for Two Candidate
  Elections}. arXiv e-prints arXiv:1902.00999 (Feb 2019)

\bibitem{Wald1945SequentialHypotheses}
Wald, A.: Sequential tests of statistical hypotheses. Ann. Math. Statist.
  \textbf{16}(2),  117--186 (June 1945). \doi{10.1214/aoms/1177731118}

\end{thebibliography}
